\documentclass[conference]{IEEEtran}
\IEEEoverridecommandlockouts
\usepackage{amsmath,amssymb,amsfonts}
\usepackage{cite}
\usepackage{graphicx}
\usepackage{epstopdf}

\usepackage{booktabs}
\usepackage{graphicx}
\usepackage{textcomp}
\usepackage{xcolor}
\usepackage{times}
\usepackage{subfigure}
\usepackage{amssymb,amsmath}
\usepackage{acronym}  % make an acronym
\usepackage{balance}

\usepackage{bm}         
\usepackage{algorithm} 
\usepackage{algorithmic}

\usepackage{lipsum}
\usepackage{stfloats}

\usepackage{url}
\usepackage{mathrsfs}
\usepackage{bbm}
\usepackage{amsthm}
\usepackage{enumerate}

\usepackage{hyperref}
\hypersetup{hidelinks, 
colorlinks=true,
allcolors=black,
pdfstartview=Fit,
breaklinks=true}

\theoremstyle{definition}

\newtheorem{lemma}{\bf Lemma}
\newtheorem{corollary}{\bf Corollary}

\theoremstyle{remark}
\newtheorem{remark}{\bf Remark}

\acrodef{ofdm}[OFDM]{orthogonal frequency division multiplexing}%
\acrodef{miso-ofdm}[MISO-OFDM]{multi-input single-output orthogonal frequency division multiplexing}%

\acrodef{ris}[RIS]{reconfigurable intelligent surface}%
\acrodef{qos}[QoS]{quality of service}%
\acrodef{idft}[IDFT]{inverse discrete Fourier transform}%
\acrodef{dft}[DFT]{discrete Fourier transform}%
\acrodef{cp}[CP]{cyclic prefix}%
\acrodef{csi}[CSI]{channel state information}%
\acrodef{awgn}[AWGN]{additive white Gaussion noise}%

\acrodef{qcqp}[QCQP]{quadratically constrained quadratic program}%
\acrodef{qp}[QP]{quadratic program}%
\acrodef{bs}[BS]{base station}%
\acrodef{ap}[BS]{base station}%
\acrodef{aps}[APs]{access points}%
\acrodef{qos}[QoS]{quality of service}%
\acrodef{ue}[UE]{user equipment}%
\acrodef{snr}[SNR]{signal-to-noise ratio}%
\acrodef{mmwave}[mmWave]{millimeter-wave}%
\acrodef{snr}[SNR]{signal-to-noise ratio}%

\acrodef{sinr}[SINR]{signal-to-interference-plus-noise ratio}%

\acrodef{ser}[SER]{symbol error rate}%
\acrodef{rc}[RC]{reflection coefficient}%
\acrodef{uavs}[UAVs]{unmanned aerial vehicles}%
\acrodef{mimo}[MIMO]{multiple-input multiple-output}%
\acrodef{noma}[NOMA]{non-orthogonal multiple access}%

\acrodef{ace}[ACE]{adaptive cross-entropy}%
\acrodef{wsr}[WSR]{weighted sum-rate}%
\acrodef{udn}[UDN]{ultra-dense network}%
\acrodef{Udn}[UDN]{Ultra-dense network}%

\def\BibTeX{{\rm B\kern-.05em{\sc i\kern-.025em b}\kern-.08em
    T\kern-.1667em\lower.7ex\hbox{E}\kern-.125emX}}

\begin{document}
\title{Location Division Multiple Access for\\ Near-Field Communications}
\author{
    \vspace{0.2cm}
    \IEEEauthorblockN{
        Zidong~Wu
        and
        Linglong~Dai\\}
    \IEEEauthorblockA{\IEEEauthorrefmark{0}
        Department of Electronic Engineering, Tsinghua University,\\
        Beijing National Research Center for Information Science and Technology (BNRist), Beijing 100084, China\\
        Email: wuzd19@mails.tsinghua.edu.cn, daill@tsinghua.edu.cn
        \vspace{-1em}
    }
}
\maketitle
\begin{abstract}
Spatial division multiple access (SDMA) is essential to improve the spectrum efficiency for multi-user multiple-input multiple-output (MIMO) communications. The classical SDMA for massive MIMO with hybrid precoding heavily relies on the angular orthogonality in the far field to distinguish multiple users at different angles, which fails to fully exploit spatial resources in the distance domain. With dramatically increasing number of antennas, extremely large-scale antenna array (ELAA) introduces additional resolution in the distance domain in the near field. In this paper, we propose the concept of location division multiple access (LDMA) to provide a new possibility to enhance spectrum efficiency. The key idea is to exploit extra spatial resources in the distance domain to serve different users at different locations (determined by angles and distances) in the near field. Specifically, the asymptotic orthogonality of beam focusing vectors in the distance domain is proved, which reveals that near-field beam focusing is able to focus signals on specific locations to mitigate inter-user interferences. Simulation results verify the superiority of the proposed LDMA over classical SDMA in different scenarios.
\end{abstract}

\vspace{-0.2cm}

\section{Introduction}
\par Massive multiple-input multiple-output (MIMO), which employs dozens or hundreds of antennas at the base station (BS), has become the key enabler to increasing spectrum efficiency by orders of magnitude in the fifth-generation (5G) networks. Spatial division multiple access (SDMA) is essential to simultaneously serve multiple user equipments (UEs) to achieve spectrum efficiency enhancement in massive MIMO systems~\cite{Marzetta'10}. Moreover, to meet the requirement of the 10-fold increase in spectrum efficiency for 6G, massive MIMO is evolving into the extremely large-scale antenna array (ELAA) equipped with thousands of antennas~\cite{Jackb'20'j}. 

% \par To meet the dramatically increasing demand for data transmission, spectrum efficiency has always been considered as one of the most important key performance indicators (KPIs) for communications. In fifth-generation (5G) networks, by employing dozens or hundreds of antennas at the base station (BS) to simultaneously serve multiple user equipments (UEs), massive multiple-input multiple-output (MIMO) is the key enabler to achieve the remarkable growth in spectrum efficiency, where spatial division multiple access (SDMA) is employed to support multi-stream transmissions with UEs~\cite{Marzetta'10}. Moreover, to meet the requirement of 10-fold increase in spectrum efficiency for 6G communications, massive MIMO technology is evolving into extremely large-scale antenna array (ELAA) equipped with thousands of antennas~\cite{Fredrik'21'j}. 

\par For the widely adopted hybrid precoding architecture in massive MIMO and ELAA systems, SDMA is usually realized by the joint design of analog and digital precoding. Since the analog precoding is realized by phase shifters, the constant modulus constraints of phase shifters impose difficulties on designs of analog precoding, which is the main challenge for SDMA in hybrid precoding schemes~\cite{Ayach'14'j}.

% \par From massive MIMO for 5G to ELAA for future 6G, people are seeking for the improvement of spectrum efficiency mainly by increasing the number of antennas to support more parallel data streams under the SDMA framework~\cite{Jackb'20'j}. To reduce the power consumption and hardware cost, hybrid precoding architectures have been widely considered~\cite{Ayach'14'j}. In hybrid precoding architectures, SDMA is usually realized by the joint design of analog and digital precoding. Since the analog precoding is realized by phase shifters, the constant modulus constraints of phase shifters impose difficulties on the analog precoding designs, which is the main challenge for SDMA in hybrid precoding schemes.

% Nevertheless, since the power consumption and hardware cost of fully-digital precoding has become unaffordable for massive MIMO and ELAA systems, the hybrid precoding architecture has been widely considered~\cite{Ayach'14'j}, where energy-efficient phase shifters are employed to realize analog precoding to reduce the number of radio frequency (RF) chains. In hybrid precoding architectures, SDMA is usually realized by the joint design of analog and digital precoders. Since the analog precoding is realized by phase shifters, the constant modulus constraints of phase shifters impose difficulties on the analog precoding designs, which is the main challenge for SDMA.
% which could approach the optimal spectrum efficiency~\cite{Raghavan'17'j}

\par Recently, owing to the characteristics of the channel, much research has focused on employing beam steering vectors to construct the analog precoder~\cite{Ayach'14'j}. By exploiting the sparsity of massive MIMO channels at high frequencies, beam steering vectors corresponding to directional beams can be directly utilized to focus the signal energy in desired directions to serve UEs. Meanwhile, owing to the angular asymptotic orthogonality of directional beams as the number of antennas tends to infinity, the received signal power could be maximized while inter-user interferences could be naturally eliminated~\cite{Xiao'15'j}. Following this insight, a practical two-stage multi-user precoding method was proposed in~\cite{Heath'15'j}. In the first stage, the analog precoder is selected from a predefined beam-steering codebook, such as discrete Fourier transform (DFT) codebooks, to maximize the received signal power and partially alleviate inter-user interferences. Then, the digital precoder is designed to further eliminate the remained interferences. 
 % . In this stage, most UEs can be distinguished and served by different directional beams, and thus the inter-user interferences can be partially alleviated.

\par Following such methods, to further enhance the spectrum efficiency, communication systems commonly rely on the high-cost way of increasing antennas to generate thinner beams, where analog precoding could eliminate interferences more thoroughly. Besides enlarging arrays only, this paper investigates a new possibility to boost the spectrum efficiency.

\par Inspired by near-field communications recently investigated in~\cite{cui'22'm}, we find that the extra resolution in the distance domain brought by near-field beams could be exploited to enhance spectrum efficiency. Specifically, the transition from massive MIMO to ELAA implies that the classical far-field \emph{planar-wave} propagation model is not accurate anymore because of the significantly increased array aperture~\cite{cui'22'm}. To precisely characterize the channel, near-field \emph{spherical-wave} propagation model has to be adopted, where the communications are referred to near-field communications. Owing to the different electromagnetic wave propagation models, unlike the far-field steering beams focusing signal energy on a certain angle, near-field beams are capable of focusing signal energy on a specific location~\cite{Heath'22'j}, which could be leveraged to mitigate interferences from UEs that can not be distinguished by far-field beams.

\par In this paper, the concept of location division multiple access (LDMA) is proposed, which aims to exploit extra spatial resources in the distance domain to enhance spectrum efficiency. Specifically, the LDMA communication scheme is proposed, which exploits the energy-focusing property of near-field beams to serve different UEs located at different angles and different distances, i.e. locations, to improve the system performance. Then, similar to the asymptotic orthogonality of far-field beams in the angular domain, the asymptotic orthogonality of near-field beams in the distance domain is investigated. Moreover, by virtue of the asymptotic orthogonality of near-field beams, the performance analysis is provided to reveal the asymptotic optimality of LDMA schemes. Simulation results are provided to verify the superiority of the LDMA scheme compared with classical SDMA schemes.

% The remainder of the paper is organized as follows. Section~\ref{sec: sys} introduces the system model. Section~\ref{sec: analysis} theoretically investigates the asymptotic orthogonality of near-field beams. In Section~\ref{sec: LDMA}, the LDMA scheme is illustrated and performance analysis is provided in Section~\ref{sec: per ana}. Simulation results are given in Section~\ref{sec: sim}, and conclusions are drawn in Section~\ref{sec: conclusion}.

\section{System Model}\label{sec: sys}
\subsection{System Model}\label{sec: sys sys}
We consider a time division duplexing (TDD) narrow-band ELAA single-cell millimeter-wave (mmWave) communication scenario. The BS is equipped with an $N$-antenna uniform linear array (ULA) and $N_{\rm{RF}}$ RF chains, where the hybrid precoding architecture is employed and $N_{\rm{RF}} \leq N$ is satisfied. The BS aims to simultaneously serve $K$ single-antenna UEs, which requires $N_{\rm{RF}} \geq K$. For analysis simplicity, $N_{\rm{RF}} = K$ is assumed. The downlink system model is first introduced and the uplink channel can be similarly obtained, which is a transpose according to the reciprocity of TDD assumption~\cite{Xiao'15'j}.

In traditional massive MIMO mmWave systems, the received signal for all $K$ UEs can be represented as
\begin{equation}
\label{eq: downlink received signal}
\begin{aligned}
{\bf{y}}^{\rm{DL}} = {\bf{H}} {\bf{F}}_{\rm{A}} {\bf{F}}_{\rm{D}} {\bf{s}} + {\bf{n}},
\end{aligned}
\end{equation}
where ${\bf{y}}^{\rm{DL}} = [y_1, y_2, \cdots, y_K]^T$ denotes the $K \times 1$ received signals for all UEs, ${\bf{H}} = [{\bf{h}}_1, {\bf{h}}_2, \cdots, {\bf{h}}_K]^H$ denotes the downlink channel, ${\bf{h}}_k$ denotes the channel vector between BS and the $k^{th}$ UE. The signal vector ${\bf{s}}$ satisfying the power constraint $\mathbb{E}[{\bf{s}}{\bf{s}}^H] = {\bf{I}}$ is transmitted to all UEs. The precoding matrices contain two components, i.e., digital precoder ${\bf{F}}_{\rm{D}}$ and analog precoder ${\bf{F}}_{\rm{A}}$. Finally, ${\bf{n}} \sim \mathcal{CN}(0, \sigma_n^2{\bf{I}})$ denotes the Gaussian noise, where $\sigma_n^2$ denotes the variance of the noise.  
\par The system spectrum efficiency could be expressed as
\begin{equation}
\label{eq: spectrum efficiency}
\begin{aligned}
R =  \sum_k R_k = \sum_k \log_2\left(1+\frac{p_k|{\bf{h}}_k^H{\bf{F}}_{\rm{A}}{\bf{f}}_{{\rm{D}},k}|^2}{\sigma_n^2 + \sum_{l \neq k} p_l|{\bf{h}}_k^H{\bf{F}}_{\rm{A}}{\bf{f}}_{{\rm{D}},l}|^2}\right),
\end{aligned}
\end{equation}
where $p_k$ denotes the power allocated to the $k^{th}$ UE, ${\bf{f}}_{{\rm{D}},k}$ denotes the $k^{th}$ column of digital precoder ${\bf{F}}_{\rm{D}}$.

\subsection{Far-Field Channel Model}\label{sec: sys channel far}
The wireless channel can be constructed by either far-field~\cite{Ayach'14'j} or near-field model~\cite{Cui'22'tcom}, where the commonly adopted boundary is Rayleigh distance $r_{\rm_{RD}} = \frac{2D^2}{\lambda}$, where $D$ and $\lambda$ denote the array aperture and wavelength, respectively~\cite{Sherman'62'j}. In 5G massive MIMO communications where the array is not very large, Rayleigh distance is only several meters, meaning that UEs are usually located in the far field, where channels can be modeled by the planar-wave propagation model as
\begin{equation}
\label{eq: far-field channel}
\begin{aligned}
{\bf{h}}_k^{\rm{far}} = \sqrt{N} \alpha_{0} {\bf{a}}(\phi_0) + \sqrt{\frac{N}{L}} \sum_{l=1}^{L}\alpha_{l} {\bf{a}}(\phi_l),
\end{aligned}
\end{equation}
which contains one line-of-sight (LoS) path and $L$ non-line-of-sight (NLoS) paths. Parameters $\alpha_l$ and $\phi_l$ denote the complex path gain and azimuth angle of the $l^{th}$ path, respectively. The index $l=0$ represents the LoS path, while $l \geq 1$ represents the NLoS paths. The channel gains $\alpha_{l}$ obey $\alpha_{l} \sim \mathcal{CN}(0, \sigma_{\alpha,l}^2)$ for each path, where $\sigma_{\alpha,0}^2 = \frac{\kappa}{\kappa+1}$ for the LoS path and $\sigma_{\alpha,l}^2 = \frac{1}{\kappa+1}$ for NLoS paths, respectively. The Rician factor $\kappa$ denotes the power ratio of LoS and NLoS paths. Due to the planar-wave propagation model, the far-field steering vector ${\bf{a}}(\phi_l)$ for an $N$-element ULA can be expressed as 
\begin{equation}
\label{eq: ula}
\begin{aligned}
{\bf{a}}(\phi) = \frac{1}{\sqrt{N}} \left[1, e^{jkd \sin\phi}, \cdots, e^{j\frac{2\pi}{\lambda}(N-1)d \sin\phi}\right]^T,
\end{aligned}
\end{equation}
where $k = \frac{2\pi}{\lambda}$ denotes the wavenumber, and $d$ denotes the spacing of adjacent antenna elements. However, since the number of antenna elements significantly increases in ELAA systems, the Rayleigh distance significantly increases. For instance, the Rayleigh distance for a $1\,{\rm{m}}$-ULA at $30$ GHz reaches $200\,{\rm{m}}$, which covers a large proportion of a cell. Therefore, the electromagnetic propagation model has to be based on spherical waves in ELAA systems~\cite{cui'22'm}.
 
\subsection{Near-Field Channel Model}\label{sec: sys channel near}
\par Based on the spherical-wave propagation model, the near-field channel can be formulated as~\cite{Cui'22'tcom}
\begin{equation}
\label{eq: near-field channel}
\begin{aligned}
{\bf{h}}_k^{\rm{near}} = \sqrt{N} \alpha_{0} {\bf{b}}(r_0, \phi_0) + \sqrt{\frac{N}{L}} \sum_{l=1}^{L}\alpha_{l} {\bf{b}}(r_l, \phi_l),
\end{aligned}
\end{equation}
where ${\bf{b}}(r_l, \phi_l)$ denotes the near-field beam focusing vector, which focuses the signal energy at the location $(r_l, \phi_l)$. To emphasize the different properties of near-field beams, we term the single-path near-field channel as \emph{beam focusing vector}, which is opposite to the classical \emph{beam steering vector} defined in~\eqref{eq: ula} in the far-field region.

\par Adopting the point scatter assumption, the near-field beam focusing vector for an $N$-element ULA shown in Fig.~\ref{img: channel ULA} is written as
\begin{equation}
\label{eq: near field response}
\begin{aligned}
{\bf{b}}(r, \phi) = \frac{1}{\sqrt{N}}\left[e^{-jk(r^{(-\widetilde{N})}-r)},\cdots,e^{-jk(r^{(\widetilde{N})}-r)}\right]^T,
\end{aligned}
\end{equation}
where $r^{(n)}$ denotes the distance between the scatterer (or UE) and the $n^{\rm th}$ antenna element, and $r$ denotes the distance between the scatterer (or UE) and the center of the array. The maximum index is defined as $\widetilde{N} = \frac{N-1}{2}$, and $N$ is assumed to be odd. The distance term $r_l^{(n)}$ can be written as
\begin{equation}
\label{eq: near field distance term}
\begin{aligned}
r_{l}^{(n)} &= \sqrt{r_l^2+n^2d^2-2ndr_l\sin\phi_l} \\
& \mathop {\approx}\limits^{(a)} r_l - nd\sin\phi_l + \frac{n^2d^2}{2r_l}\cos^2\phi_l,
% & = r_l + \psi_{r_l, \phi_l}^{(n)}
\end{aligned}
\end{equation}
where approximation (a) is derived by the second-order Taylor series expansion $\sqrt{1+x} = 1 + \frac{x}{2} - \frac{x^2}{8} + \mathcal{O}(x^3)$. As shown in~\cite{Janaswamy'17'm}, second-order expansion is usually accurate enough.

\begin{figure}[!t]
	\centering
	\setlength{\abovecaptionskip}{0.cm}
	\includegraphics[width=3in]{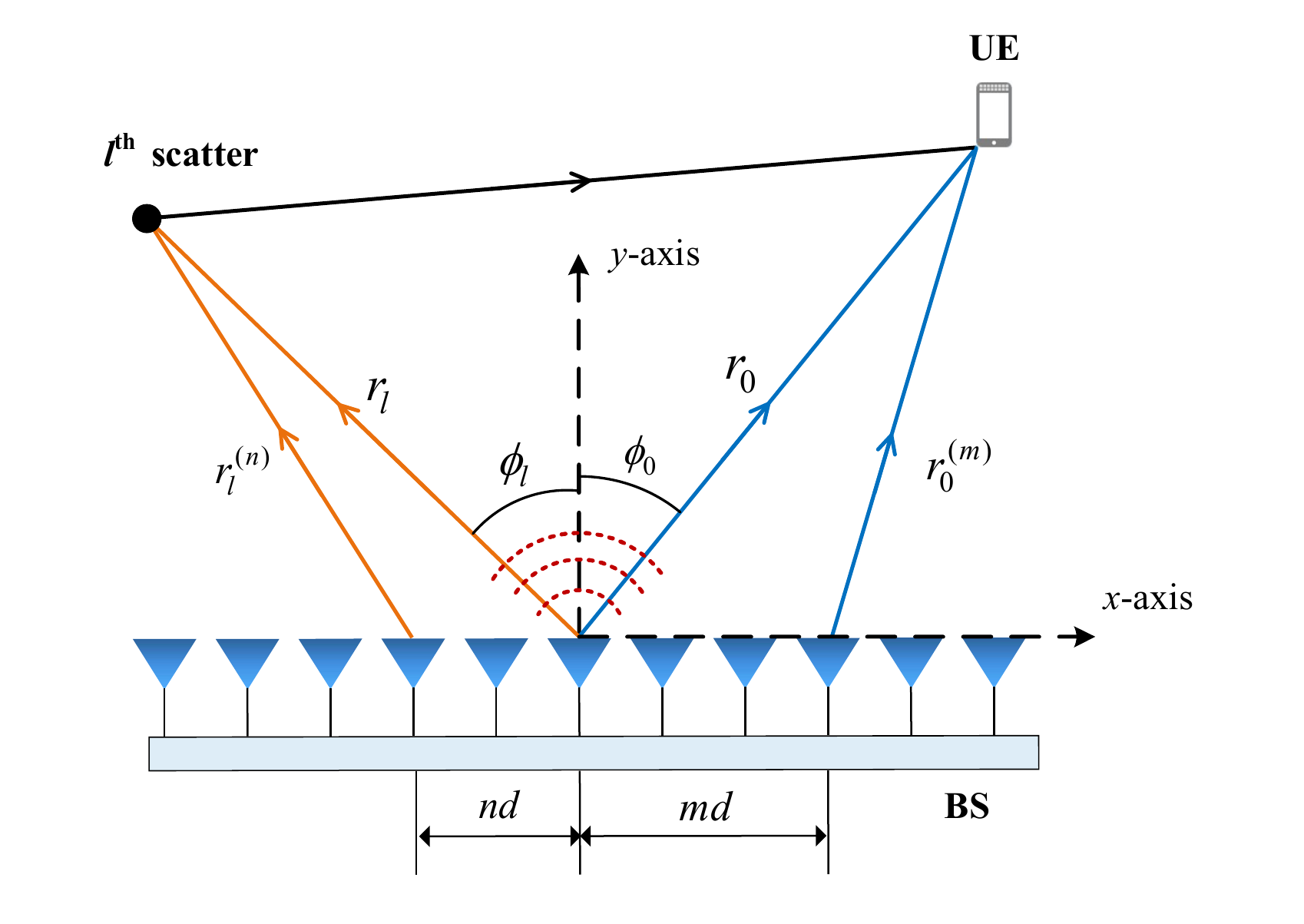}
	\caption{Near-field channel model for ULA communication systems.}
	\label{img: channel ULA}
    \vspace{-0.2cm}
\end{figure}

\begin{remark}
When Taylor series expansion only keeps the first-order term as $\sqrt{1+x} \approx 1 + \frac{x}{2}$, beam focusing vectors in~\eqref{eq: near field response} degenerate to beam steering vectors in~\eqref{eq: far-field channel}. In other words, far-field beam steering vector is a special case of near-field beam focusing vector without higher-order information.
\end{remark}

To sum up, different from far-field channels based on planar-wave assumptions in~\eqref{eq: ula}, the phase term based on~\eqref{eq: near field distance term} is related to both spatial angle and distance. Therefore, the near-field channels for UEs located at the same angle but at different distances are remarkably different. The UEs can be distinguished according to their spatial angle and distance, i.e., their location in the two-dimensional (2D) space, which is a fundamental change compared with far-field communications. In the following section, this extra resolution of near-field beams in the distance domain will be discussed.

\section{Analysis of Asymptotic Orthogonality of Near-Field Beam Focusing Vectors}\label{sec: analysis}
The correlation of classical far-field steering vectors focusing on $\phi_l$ and $\phi_m$ can be formulated as
\begin{equation}
\label{eq: far gain}
\begin{aligned}
\left|{\bf{a}}^H(\phi_l) {\bf{a}}(\phi_m)\right| & = \frac{1}{N} \left|\sum_{n=-\widetilde{N}}^{\widetilde{N}} e^{jnkd(\sin\phi_m-\sin\phi_l)}\right| \\
& = \frac{1}{N}\left| \Xi_N(kd(\sin\phi_m-\sin\phi_l)) \right|,
\end{aligned}
\end{equation}
where $\Xi_N(\alpha) = \sin\frac{N\pi}{2}\alpha/(N\sin\frac{\pi}{2}\alpha)$ is the Dirchlet sinc function. According to~\eqref{eq: far gain}, the correlation of steering vectors achieves the maximum when $\phi_l = \phi_m$. If we consider two single-path UEs located at the same angle, the correlation of their channels achieves the maximum and simultaneous transmissions can not be established through precoding. Otherwise, the correlation of steering vectors focusing on different angles tends to be orthogonal with infinite antennas as~\cite{Xiao'15'j}
\begin{equation}
\label{eq: far infinity}
\begin{aligned}
\lim_{N \to +\infty} |{\bf{a}}^H(\phi_l) {\bf{a}}(\phi_m)| = 0 \ , \  \phi_l \neq \phi_m.
\end{aligned}
\end{equation}

Therefore, the spatial angular resolution of BS tends to infinity as the number of antennas increases. Owing to the angular orthogonality, BS could distinguish different channel components and multiplex different data streams to different UEs. Therefore, the angular orthogonality of the far-field beam steering vectors contributes to the SDMA scheme \cite{Xiao'15'j}. 

\par Similarly, we wish to analyze the correlation of near-field beam focusing vectors defined in~\eqref{eq: near field response}. The correlation of two beam focusing vectors corresponding to the location of $(r_l, \phi_l)$ and $(r_m, \phi_m)$ is written as $\left|{\bf{b}}^H(r_l, \phi_l) {\bf{b}}(r_m, \phi_m)\right|$.
% \begin{equation}
% \label{eq: near gain}
% \begin{aligned}
% \left|{\bf{b}}^H(r_l, \phi_l) {\bf{b}}(r_m, \phi_m)\right| = \frac{1}{N} \left| \sum_{n=-\widetilde{N}}^{\widetilde{N}} e^{jk (\psi_{r_l, \phi_l}^{(n)} - \psi_{r_m, \phi_m}^{(n)})} \right|.
% \end{aligned}
% \end{equation}
According to {\bf{Lemma} 1} in~\cite{Cui'22'tcom}, the correlation of near-field beam focusing vectors corresponding to the same angle but different distances can be illustrated with the following lemma. 
\begin{lemma}
\label{lemma1}
The correlation of near-field beam focusing vectors corresponding to different distances can be approximated as
\begin{equation}
\label{eq: near field correlation}
\begin{aligned}
\left|{\bf{b}}^H(r_l, \phi) {\bf{b}}(r_m, \phi) \right| \approx \left|G(\beta)\right| =  \left| \frac{C(\beta)+jS(\beta)}{\beta} \right|, 
\end{aligned}
\end{equation}
where $\beta = N\sqrt{\frac{d^2\cos^2\phi}{2\lambda}|\frac{1}{r_l}-\frac{1}{r_m}|}$, $C(\cdot)$ and $S(\cdot)$ denote the Fresnel functions written as $C(x) = \int_0^x \cos(\frac{\pi}{2}t^2){\rm{d}}t$ and $S(x) = \int_0^x \sin(\frac{\pi}{2}t^2){\rm{d}}t$~\cite{Sherman'62'j}.
\end{lemma}
This lemma reveals that the correlation of near-field beam focusing vectors varies with the distance, which fundamentally differs from far-field regions where the correlation of two far-field beam steering vectors is invariable over different distances. We show that the correlation tends to zero as the number of antennas scales up as follows. 

\begin{corollary}[Asymptotic Orthogonality in Distance Domain]
\label{coro1}
Near-field beam focusing vectors corresponding to the same angle and different distances are asymptotically orthogonal with the increasing number of antennas, which is to say
\begin{equation}
\label{eq: near lim}
\begin{aligned}
\lim_{N \to +\infty} \left|{\bf{b}}^H(r_l, \phi) {\bf{b}}(r_m, \phi) \right|  = 0,\ {\rm{for}}\ r_l \neq r_m.
\end{aligned}
\end{equation}
\end{corollary}
\begin{proof}
As shown in~\eqref{eq: near field correlation}, when the number of antennas $N$ tends to infinity, the numerator $C(\beta) + jS(\beta)$ converges to $0.5 + 0.5j$ and the denominator $\beta$ tends to $+\infty$~\cite{Sherman'62'j}. Therefore, the correlation converges to 0, which proves the asymptotic orthogonality of beam focusing vectors.
\end{proof}
To verify the orthogonality in the distance domain, the correlation of beam focusing vectors corresponding to $(5\,{\rm{m}}, \pi/6)$ and $(15\,{\rm{m}}, \pi/6)$ is plotted in Fig.~\ref{img:Cor_ula}. It is shown that the correlation significantly decreases as the number of antennas increases. Moreover,~\eqref{eq: near field correlation} can well approximate the accurate correlation. Furthermore, we prove a more general 2D asymptotic orthogonality in both angular and distance domains with the following corollary.
% \footnote{It is worth noting that the blue solid line in Fig.~\ref{img:Cor_ula} is obtained with the approximation of $r_l$ in~\eqref{eq: near field distance term}. One may notice that when the number of antennas tends to infinity, the near-field assumption does not hold anymore and the approximation in~\eqref{eq: near field distance term} is no longer valid since higher-order terms of Taylor series occur. However, we mainly focus on the asymptotic trends assuming that near-field assumptions are valid, which is similar to the asymptotic orthogonality in the angular domain in the far-field region \cite{Xiao'15'j}.}

\begin{figure}[!t]
	\centering
	\setlength{\abovecaptionskip}{0.cm}
	\includegraphics[width=3in]{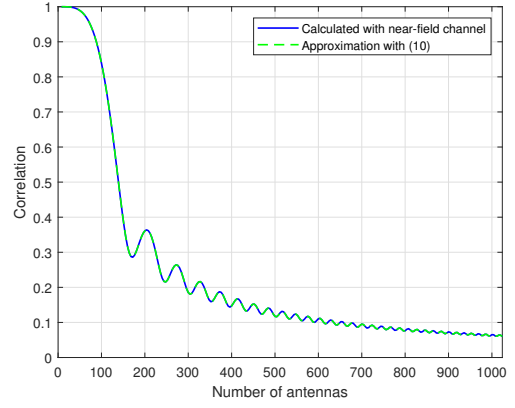}
	\caption{Correlation of beam focusing vectors with increasing antennas. The frequency is set to $30$ GHz and the antennas are half-wavelength spaced.}
	\label{img:Cor_ula}
    \vspace{-0.2cm}
\end{figure}

\begin{corollary}[Asymptotic Orthogonality in 2D Domain]
\label{coro2}
Near-field beam focusing vectors corresponding to any different angles or different distances are also asymptotically orthogonal with the increasing number of antennas, which is to say
\begin{equation}
\label{eq: near lim 2D}
\begin{aligned}
\lim_{N \to +\infty} \left|{\bf{b}}^H(r_l, \phi_l) {\bf{b}}(r_m, \phi_m)\right| = 0,\,{\rm{for}}\,r_l \neq r_m\ {\rm{or}}\ \phi_l \neq \phi_m.
\end{aligned}
\end{equation}
\end{corollary}

\begin{proof}
The proof can be seen in~\cite{Zidong'22'jsac} Appendix B.
\end{proof}
\par The corollary reveals that the angular orthogonality in the far-field region generalizes into the 2D orthogonality in the near-field region, which indicates a stronger potential to establish simultaneous transmissions through precoding in ELAA communications. Therefore, it brings possibilities for a novel multiple access scheme, which is discussed as follows.

\section{Near-Field LDMA Scheme}\label{sec: LDMA}
\par As discussed above, near-field beamforming is capable of focusing the energy on a specific location rather than a specific angle. It indicates that, far-field steering beams can be replaced by near-field focusing beams to distinguish users at different locations and suppress interferences from other UEs in the same direction. Following this intuition, the concept of LDMA is proposed, employing near-field location-dependent beam focusing vectors as analog precoders to serve different UEs located in different locations. Compared with SDMA, the proposed LDMA provides a method to harvest extra orthogonal resources in the distance dimension.

\par Specifically, the proposed LDMA scheme comprises three main stages, including: Initial access, uplink equivalent channel estimation, and uplink/downlink data transmission. In the initial access procedure, BS could perform beam sweeping to establish a physical link connection to idle UEs. Each UE is scheduled with a specific codeword as the analog precoder ${\bf{w}}_k$ from a predefined near-field codebook $\mathcal{W} = [w_1, \cdots, w_M]$. It is worth noting that, the codebook can be designed with constraints on the maximum correlation between different codewords~\cite{Cui'22'tcom}. For fixed correlation, the distance between focal points can be determined by~\eqref{eq: near field correlation} correspondingly, revealing the distance resolution of LDMA. Therefore, by adjusting the correlation of codewords, the distance resolution can be determined. Then, the analog precoding for all connected UEs can be designed as ${\bf{F}}_{\rm{A}} = [w_1, w_2, \cdots, w_K]$. Afterward, BS could estimate the effective channel with uplink pilots as
\begin{equation}
\label{eq: effective channel}
\begin{aligned}
{\overline{\bf{h}}}_{k} = {\bf{F}}_{\rm{A}}^H {\bf{h}}_k + {\bf{F}}_{\rm{A}}^H n_k,
\end{aligned}
\end{equation}
where $n_k$ denotes the noise corresponding to the $k^{th}$ UE.
Finally, the BS could design the digital precoder through the estimated effective channel by weighted minimum mean square error (WMMSE) as in \cite{He'11'j} or zero-forcing (ZF) as
\begin{equation}
\label{eq: zero forcing}
\begin{aligned}
{\bf{F}}_{\rm{D}} = {\overline{\bf{H}}}^H \left({\overline{\bf{H}}}\ {\overline{\bf{H}}}^H\right)^{-1} {\bf{\Lambda}},
\end{aligned}
\end{equation}
where ${\overline{\bf{H}}} = [{\overline{\bf{h}}}_1, \cdots, {\overline{\bf{h}}}_K]^H$ denotes the effective channel of all $K$ UEs. The diagonal matrix $\Lambda$ denotes the power allocation for different UEs, which is designed to satisfy $\|{\bf{F}}_{\rm{A}} {\bf{f}}_{{\rm{D}},k}\|^2 = 1$. Apart from WMMSE and ZF, other design methods are also supported to design digital precoders. To sum up, the LDMA scheme can be summarized in {\bf{Algorithm}~\ref{alg:2}}.
\addtolength{\topmargin}{0.01in}
\begin{algorithm}[t]
	\caption{Location Division Multiple Access.} 
	\label{alg:2} 
	\begin{algorithmic}[1] %这个1 表示每一行都显示数字
		\REQUIRE ~ %算法的输入参数：Input
		Multi-user channel ${\bf{H}}$, near-field codebook $\mathcal{W}$.
		\ENSURE ~ %算法的输出：Output
        Digital precoder ${\bf{F}}_{\rm{D}}$ and analog precoder ${\bf{F}}_{\rm{A}}$ 
		% \STATE {\bf{Codebook Design:}}
        % \STATE Construct the near-field codebook $\mathcal{W}$;
        \STATE {\bf{Initial Access:}} 
        \STATE BS performs beam sweeping with $\mathcal{W}$ and each UE reports the beam decision to BS;
		\STATE BS selects the best codeword ${\bf{w}}_k$ for $k^{\rm th}$ UE from $\mathcal{W}$ and construct the analog precoder ${\bf{F}}_{\rm{A}} = [{\bf{w}}_1, {\bf{w}}_2, \cdots, {\bf{w}}_K]$;
        \STATE {\bf{Uplink Equivalent Channel Estimation:}} 
        \STATE Each UE sends uplink non-orthogonal pilots;
        \STATE BS estimates the effective channel ${\overline{\bf{h}}}_{k}$ by~\eqref{eq: effective channel};
        \STATE BS designs digital combiner and precoder ${\bf{F}}_{\rm{D}}$ by~\eqref{eq: zero forcing};
        \STATE {\bf{Uplink/Downlink Data Transmission:}}
		\STATE BS performs combining/precoding with ${\bf{F}}_{\rm{D}}$ and ${\bf{F}}_{\rm{A}}$;
		\RETURN Digital precoder ${\bf{F}}_{\rm{D}}$ and analog precoder ${\bf{F}}_{\rm{A}}$. %算法的返回值
	\end{algorithmic}
\end{algorithm}

\section{Performance Analysis}\label{sec: per ana}
% In this section, the performance analysis of LDMA transmission is provided. We first show that the asymptotic spectrum efficiency of proposed LDMA could achieve the ideal capacity without multi-user interferences. Then, we explore the performance gain under a special scenario where multiple UEs are located within the same direction, which reveals that LDMA could elaborately utilize the spatial resources to enhance system performance compared with classical SDMA.
\subsection{Asymptotic Spectrum Efficiency for Single-path Channels}\label{sec: asy single}
To investigate the system performance of LDMA, we first assume a single-path channel scenario. The single-path channel could be written as ${\bf{h}}_k = \sqrt{N} \alpha_k {\bf{b}}(r_k, \phi_k)$. If BS acquires the perfect channel state information, the optimal analog precoder for all $K$ UEs can be written as ${\bf{F}}_{\rm{A}} = {\bf{B}} = [{\bf{b}}(r_1, \phi_1), \cdots, {\bf{b}}(r_K, \phi_K)]$. For analysis simplicity, the ZF-based digital precoder is adopted. In addition, the large-scale fading is neglected by employing reasonable power control and thus different UEs share the same channel gain. The spectrum efficiency can be obtained by the following lemma.

\begin{lemma}
\label{lemma5}
With equal power allocation for different UEs, the spectrum efficiency achieved by {\bf{Algorithm}~\ref{alg:2}} is given by 
\begin{equation}
\label{eq: EE single}
\begin{aligned}
R = \sum_{k=1}^K \log_2 \left(1+\frac{P}{K\sigma_n^2} \frac{N |\alpha_k|^2}{[{\bf{B}}^H{\bf{B}}]_{k,k}^{-1}} \right),
\end{aligned}
\end{equation}
where $P$ denotes the total transmission power, $[{{\bf{B}}^{H}}{\bf{B}}]_{k,k}^{-1}$ denotes the $k^{\rm th}$ diagonal entry of the matrix $({{\bf{B}}^{H}}{\bf{B}})^{-1}$. 
\end{lemma}
\begin{proof}
The proof can be seen in~\cite{Zidong'22'jsac} Appendix C.
\end{proof}
For multi-user MIMO systems, an ideal communication scenario is that multiple UEs could be served by BS without interferences, leading to the ideal spectrum efficiency as
\begin{equation}
\label{eq: capacity opt}
\begin{aligned}
\hat{R} = \sum_{k=1}^K \log_2 \left(1+\frac{P}{K\sigma^2} N|\alpha_k|^2 \right).
\end{aligned}
\end{equation}
Due to the asymptotic orthogonality proved in {\bf Corollary~\ref{coro2}}, the spectrum efficiency could approach the ideal one according to the following corollary.

\begin{corollary}
\label{coro3}
If the number of BS antennas $N$ tends to infinity, the spectrum efficiency for a fixed number of UEs approaches the ideal spectrum efficiency with probability one, i.e., 
\begin{equation}
\label{eq: asymptotic spectrum efficiency}
\begin{aligned}
\mathbb{P}\left[ \lim_{N \to +\infty} R = \hat{R}\right] = 1.
\end{aligned}
\end{equation}
\end{corollary}

\begin{proof}
The proof can be seen in~\cite{Zidong'22'jsac} Appendix D.
\end{proof}
Therefore, as the number of antennas tends to infinity, the asymptotic optimality could be ensured for LDMA schemes.

\subsection{Linear Distribution Analysis}\label{sec: per B}
Compared with SDMA, the ability to serve single-path UEs residing in the same angle could be viewed as a feature distinguishing LDMA from SDMA. To verify this advantage, we consider a special scenario where UEs are linearly distributed.

First, we consider a simple three UEs system, where UEs are distributed in a certain direction $\phi$ and distance range of $[r_{\rm{min}}, r_{\rm{max}}]$. Without loss of generality, we assume $r_1 \leq r_2 \leq r_3$. According to {\bf{Lemma}~\ref{lemma5}}, the analog precoder can be obtained as ${\bf{B}}_{\rm{TU}} = [{\bf{b}}(r_1, \phi),{\bf{b}}(r_2, \phi), {\bf{b}}(r_3, \phi)]$. The maximum spectrum efficiency of this system can be obtained based on the following assumptions:
\begin{enumerate}[(i)]
\item Single-path channel is adopted and high SNR is assumed.
\item Interferences of non-adjacent UEs are neglected, i.e. the interference of $1^{\rm st}$ and $3^{\rm rd}$ UE is approximated by 0.
\item ZF-based digital precoders are adopted.
\end{enumerate}
\par Then, we can denote ${\bf{B}}_{\rm{TU}}^H {\bf{B}}_{\rm{TU}}$ as
\begin{equation}
\label{eq: T three users}
{\bf{T}}_{\rm{TU}} = {\bf{B}}_{\rm{TU}}^H {\bf{B}}_{\rm{TU}} = 
\left[
\begin{array}{ccc}
1 & \delta_{21}^* & 0 \\
\delta_{21} & 1 & \delta_{32}^* \\
0 & \delta_{32} & 1
\end{array}
\right],
\end{equation}
where $\delta_{ij} = {\bf{b}}^H(r_i, \phi) {\bf{b}}(r_j, \phi)$ denotes the product of the $i^{\rm th}$ and $j^{\rm th}$ beam focusing vector. Suppose  the locations of the first and third UE are fixed, the maximum achievable spectrum efficiency could be approximated by the following lemma.
\begin{lemma}
\label{lemma6}
We assume $\delta_{21} = g(x)$ and $\delta_{32} = g(r_0 - x)$, where $g(x)$ is defined as $g(x): \mathbb{R}_+ \to \mathbb{C}$, and $|g(x)|^2$ is monotonically decreasing and convex, satisfying $0 \leq |g(x)| \leq 1$ and $r_0>0$. The spectrum efficiency with ${\bf{T}}_{\rm{TU}}$ satisfies 
\begin{equation}
\label{eq: three user max SE}
\begin{aligned}
R_{\rm TU}(x) \lesssim R(x)^{\rm aub} = 2\log_2 \left(1+\frac{P}{K\sigma_n^2} \frac{N |\alpha|^2 (1-2g(\hat{x}))}{(1-g(\hat{x}))} \right) \\
+ \log_2 \left(1+\frac{P}{K\sigma_n^2} N |\alpha|^2 (1-2g(\hat{x})) \right).
\end{aligned}
\end{equation}
The equality holds when $x = \hat{x}$ satisfying $g(\hat{x}) = g(r_0-\hat{x})$.
\end{lemma}
\begin{proof}
The proof can be seen in~\cite{Zidong'22'jsac} Appendix E.
\end{proof}
\begin{remark}
It can be proved that the envelope of $|G(\beta)|$ in~\eqref{eq: near field correlation} approximately satisfies the monotonically decreasing and convex property. Therefore, the maximum spectrum efficiency can be approximately obtained through {\bf{Lemma}~\ref{lemma6}} for a three-UE system. It is worth noting that, this conclusion is drawn under the assumption that interference from non-adjacent UEs is neglected. Therefore, this lemma can be viewed as an approximated upper bound for real communication scenarios.
\end{remark}
\par Under the same three assumptions, the conclusion of {\bf{Lemma}~\ref{lemma6}} can be generalized into multi-user scenarios. Specifically, we consider a system where multiple UEs are linearly distributed. If we investigate any adjacent three UEs, the middle UE must be located in the position according to~{{\bf Lemma~\ref{lemma6}}}. Then, the maximum spectrum efficiency of linearly distributed UEs can be obtained through the following lemma.

% Under the same three assumptions, the same conclusion of {\bf{Lemma}~\ref{lemma6}} can be generalized into multi-user scenarios. Specifically, we consider a system where multiple UEs are linearly distributed. If we investigate any adjacent three UEs, the middle UE indexed by $i$ must be located in the position that makes the same interference on the two adjacent UEs to maximize the spectrum efficiency, which is to say $|{\bf{b}}_i^H{\bf{b}}_{i+1}| = |{\bf{b}}_i^H{\bf{b}}_{i-1}|$ for $i=2,3,\cdots,K-1$. Otherwise, the system spectrum efficiency can be further enhanced by adjusting the location of the middle UE. Then, when $K$ UEs are placed satisfying this condition, the maximum spectrum efficiency of multiple UEs linearly distributed can be obtained through the following lemma.
\begin{lemma}
\label{lemma7}
The approximated upper bound of the spectrum efficiency of linearly distributed UEs can be expressed as
\begin{equation}
\label{eq: multiple users maximum}
\begin{aligned}
R_{K} \lesssim  R^{\rm{aub}} = \sum_{k=1}^K \log_2 \left(1+\frac{P}{K\sigma_n^2} \frac{N |\alpha_k|^2}{\gamma_k} \right),
\end{aligned}
\end{equation}
where $\gamma_k$ is determined by
\begin{equation}
\label{eq: gamma formulation}
\begin{aligned}
\gamma_k = \frac{(\chi_1 x_1^{k-2}+\chi_2x_2^{k-2})(\chi_1x_1^{K-k-1}+\chi_2x_2^{K-k-1})}{\chi_1x_1^{K-1}+\chi_2x_2^{K-1}},
\end{aligned}
\end{equation}
and $x_1, x_2$ are solutions to $x^2-x+|\delta|^2=0$, and $\chi_1 = -\frac{x_1^2}{x_2-x_1}$, $\chi_2 = \frac{x_2^2}{x_2-x_1}$, $|\delta|$ is defined as the min-max correlation between UEs as $|\delta| = \mathop{\min}\limits_{r_1,\cdots,r_K} \mathop{\max}\limits_{i\neq j} |{\bf{b}}^H(r_i,\phi){\bf{b}}(r_j,\phi))|$.
\end{lemma}
\begin{proof}
The proof can be seen in~\cite{Zidong'22'jsac} Appendix F.
% According to the conclusion of the inversion of a tridiagonal matrix \cite{usmani'1994'}, the diagonal elements of ${\bf{T}}$ could be obtained. According to {\bf{Lemma}~\ref{lemma5}}, the right side of~\eqref{eq: multiple users maximum} could be obtained.
\end{proof}

\section{Simulation Results}\label{sec: sim}
\subsection{Linear Distribution Scenarios}\label{sec: sim linear}
We first consider a linear distribution scenario, where UEs are aligned along the spatial angle $\phi = 0$. Each UE is located within the range $[4\,{\rm{m}}, 150\,{\rm m}]$. The number of UEs varies in $[1, 14]$ and SNR is set to be $12\, {\rm{dB}}$. The frequency is $30$ GHz and the array is half-wavelength spaced $256$-element ULA with $d = \lambda/2 = 0.5\,{\rm cm}$. 
\par The spectrum efficiency with increasing number of UEs is plotted in Fig.~\ref{img: linear maximum SE}. The red line and blue line represent the situation where UEs are placed to minimize the interference from adjacent UEs without and with non-adjacent (NA) interference, respectively. The red line denotes the approximated upper bound derived in~\eqref{eq: multiple users maximum} neglecting the NA interference, which is unreachable. A reachable spectrum efficiency for the same UE position is plotted in blue. An exhaustive search is performed to search for the realistic maximum spectrum efficiency by changing the position of UEs, which is plotted in green. It shows that the approximated upper bound in~\eqref{eq: multiple users maximum} is very tight for small number of UEs and provides an approximately ideal distribution of UEs. Therefore, {\bf{Lemma}~\ref{lemma6}} provides accurate estimations of maximum spectrum efficiency.
\par In addition, the orange line represents randomly and linearly distributed UEs and the black dashed line denotes the spectrum efficiency employing far-field SDMA where only one UE could access to the BS since channels for all UEs are constructed based on the same steering vector. It shows that even randomly distributed UEs with LDMA also outperforms far-field SDMA when the number of UEs is not very large.
\begin{figure}[!t]
	\centering
	\setlength{\abovecaptionskip}{0.cm}
	\includegraphics[width=3in]{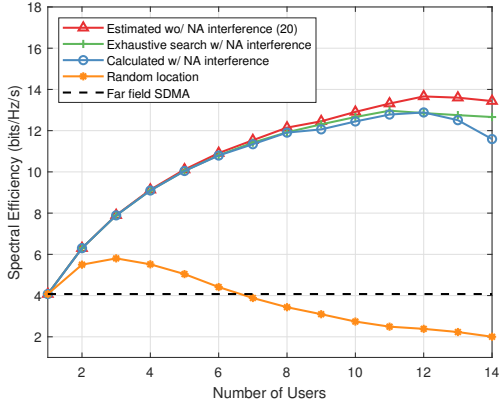}
	\caption{Verification of the approximated upper bound of the spectrum efficiency for linear distributed UEs.}
	\label{img: linear maximum SE}
    \setlength{\belowcaptionskip}{-1cm}
    \vspace{-0.3cm}
\end{figure}

\par Then, we consider a multi-path model where both LoS and NLoS paths exist. The UEs are located within the range of $[4\,{\rm{m}}, 100\,{\rm m}]$. The number of UEs is $K=4$ and the number of NLoS paths is $L=5$. The near-field polar-domain codebook in~\cite{Cui'22'tcom} is adopted and the far-field DFT codebook is adopted as the classical far-field SDMA baseline~\cite{Heath'15'j}. Both ZF and WMMSE are considered for designing digital precoders. The simulation result is shown in Fig.~\ref{img: line ULA}. The proposed LDMA outperforms SDMA on both WMMSE and ZF scenarios, with about $60\%$ and $240\%$ performance gain compared with WMMSE- and ZF-based SDMA at SNR $=\,20$ dB, respectively. In addition, the proposed WMMSE-based LDMA outperform the ZF-based fully-digital precoding in low SNR scenarios since ZF enlarges the noise. In high SNR scenarios, WMMSE-based LDMA could approach the ideal fully-digital precoding.

\begin{figure}[!t]
	\centering
	\setlength{\abovecaptionskip}{0.cm}
	\includegraphics[width=3in]{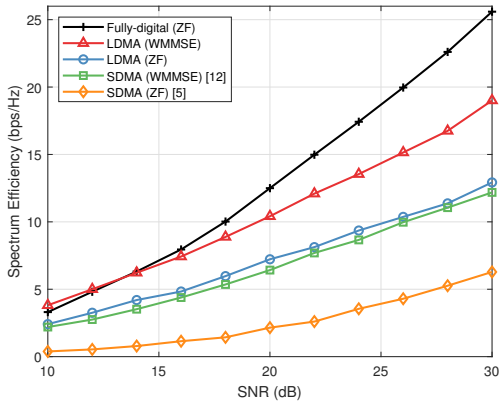}
	\caption{Comparison of the proposed LDMA and classical far-field multiple access under the linear distribution assumption.}
	\label{img: line ULA}
    \setlength{\belowcaptionskip}{-2cm}
    \vspace{-0.3cm}
\end{figure}

\vspace{-0.1cm}
\subsection{Uniform Distribution Scenarios}\label{sec: sim uniform}
\par Next, we consider a more general scenario where UEs are uniformly distributed within the cell, with a radius range $[4\,{\rm m}, 100\,{\rm m}]$ and angle range $[-\pi/3$, $\pi/3]$. The number of NLoS paths and UEs are set to $L=5$ and $K = 10$. Both ZF and WMMSE are adopted for designing digital precoders. Simulation results are shown in Fig.~\ref{img: uniform ULA}. The WMMSE-based LDMA scheme achieves nearly $100\%$ improvement at SNR $=\,20$ dB compared with WMMSE-based SDMA. Also, since the near-field precoding could be leveraged to manage inter-user interference, the performance gap between proposed LDMA with far-field SDMA increases as SNR increases. 

\begin{figure}[!t]
	\centering
	\setlength{\abovecaptionskip}{0.cm}
	\includegraphics[width=3in]{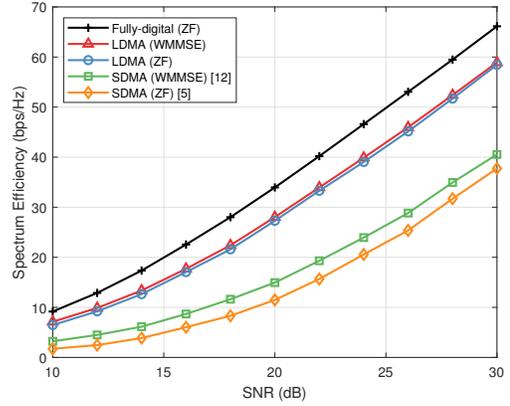}
	\caption{Comparison of the proposed LDMA and classical SDMA under the uniform distribution assumption.}
	\label{img: uniform ULA}
    \setlength{\belowcaptionskip}{-5cm}
\end{figure}

\vspace{-0.2cm}
\section{Conclusions}\label{sec: conclusion}
In this paper, the concept of LDMA is first proposed, which leverages the near-field beam focusing property to mitigate interferences and significantly enhance spectrum efficiency. Similar to the asymptotic angular orthogonality of far-field beams, the asymptotic orthogonality of near-field beams in the distance domain is proved and the LDMA scheme is investigated. Simulation results verify the superiority of the proposed LDMA scheme over uniform and linear distribution scenarios. We hope the LDMA scheme could provide a new possibility for spectrum efficiency enhancement compared classical SDMA for the widely adopted hybrid precoding.
\footnotesize
\balance 
\bibliographystyle{IEEEtran}
\bibliography{IEEEabrv,reference}

\begin{thebibliography}{10}
\providecommand{\url}[1]{#1}
\csname url@samestyle\endcsname
\providecommand{\newblock}{\relax}
\providecommand{\bibinfo}[2]{#2}
\providecommand{\BIBentrySTDinterwordspacing}{\spaceskip=0pt\relax}
\providecommand{\BIBentryALTinterwordstretchfactor}{4}
\providecommand{\BIBentryALTinterwordspacing}{\spaceskip=\fontdimen2\font plus
\BIBentryALTinterwordstretchfactor\fontdimen3\font minus
  \fontdimen4\font\relax}
\providecommand{\BIBforeignlanguage}[2]{{%
\expandafter\ifx\csname l@#1\endcsname\relax
\typeout{** WARNING: IEEEtran.bst: No hyphenation pattern has been}%
\typeout{** loaded for the language `#1'. Using the pattern for}%
\typeout{** the default language instead.}%
\else
\language=\csname l@#1\endcsname
\fi
#2}}
\providecommand{\BIBdecl}{\relax}
\BIBdecl

\bibitem{Marzetta'10}
T.~L. {Marzetta}, ``Noncooperative cellular wireless with unlimited numbers of
  base station antennas,'' \emph{{IEEE} Trans. Wireless Commun.}, vol.~9,
  no.~11, pp. 3590--3600, Oct. 2010.

\bibitem{Jackb'20'j}
L.~{Sanguinetti}, E.~{Björnson}, and J.~{Hoydis}, ``Toward massive {MIMO} 2.0:
  Understanding spatial correlation, interference suppression, and pilot
  contamination,'' \emph{{IEEE} Trans. Commun.}, vol.~68, no.~1, pp. 232--257,
  Jan. 2020.

\bibitem{Ayach'14'j}
O.~{Ayach}, S.~{Rajagopal}, S.~{Abu-Surra}, Z.~{Pi}, and R.~W. {Heath},
  ``Spatially sparse precoding in millimeter wave {MIMO} systems,''
  \emph{{IEEE} Trans. Wireless Commun.}, vol.~13, no.~3, pp. 1499--1513, Jan.
  2014.

\bibitem{Xiao'15'j}
C.~{Sun}, X.~{Gao}, S.~{Jin}, M.~{Matthaiou}, Z.~{Ding}, and C.~{Xiao}, ``Beam
  division multiple access transmission for massive {MIMO} communications,''
  \emph{{IEEE} Trans. Commun.}, vol.~63, no.~6, pp. 2170--2184, Apr. 2015.

\bibitem{Heath'15'j}
A.~{Alkhateeb}, G.~{Leus}, and R.~W. {Heath}, ``Limited feedback hybrid
  precoding for multi-user millimeter wave systems,'' \emph{{IEEE} Trans.
  Wireless Commun.}, vol.~14, no.~11, pp. 6481--6494, Jul. 2015.

\bibitem{cui'22'm}
M.~{Cui}, Z.~{Wu}, Y.~{Lu}, X.~{Wei}, and L.~{Dai}, ``Near-field communications
  for {6G}: Fundamentals, challenges, potentials, and future directions,''
  \emph{{IEEE} Commun. Mag. (early access)}, pp. 1--7, Sep. 2022.

\bibitem{Heath'22'j}
N.~J. {Myers} and R.~W. {Heath}, ``Infocus: A spatial coding technique to
  mitigate misfocus in near-field {LoS} beamforming,'' \emph{{IEEE} Trans.
  Wireless Commun.}, vol.~21, no.~4, pp. 2193--2209, Sep. 2022.

\bibitem{Cui'22'tcom}
M.~{Cui} and L.~{Dai}, ``Channel estimation for extremely large-scale {MIMO}:
  Far-field or near-field?'' \emph{{IEEE} Trans. Commun.}, vol.~70, no.~4, pp.
  2663--2677, Jan. 2022.

\bibitem{Sherman'62'j}
J.~{Sherman}, ``Properties of focused apertures in the fresnel region,''
  \emph{{IEEE} Trans. Antennas Propag.}, vol.~10, no.~4, pp. 399--408, Jul.
  1962.

\bibitem{Janaswamy'17'm}
K.~T. {Selvan} and R.~{Janaswamy}, ``Fraunhofer and fresnel distances: Unified
  derivation for aperture antennas,'' \emph{{IEEE} Antennas Propag. Mag.},
  vol.~59, no.~4, pp. 12--15, Jun. 2017.

\bibitem{Zidong'22'jsac}
Z.~{Wu} and L.~{Dai}, ``Multiple access for near-field communications: {SDMA}
  or {LDMA}?'' \emph{arXiv preprint arXiv:2208.06349}, Oct. 2022.

\bibitem{He'11'j}
Q.~{Shi}, M.~{Razaviyayn}, Z.-Q. {Luo}, and C.~{He}, ``An iteratively weighted
  {MMSE} approach to distributed sum-utility maximization for a {MIMO}
  interfering broadcast channel,'' \emph{{IEEE} Trans. Signal Process.},
  vol.~59, no.~9, pp. 4331--4340, Apr. 2011.

\end{thebibliography}
\end{document}